%% file: arxiv-version.tex
\documentclass[a4paper,11pt]{scrartcl}

\usepackage[utf8]{inputenc}
\usepackage[breaklinks=true]{hyperref}
\usepackage[anythingbreaks]{breakurl}

\usepackage{amsthm}
\usepackage{amsmath}
\usepackage{amssymb}
\usepackage{enumerate}
\usepackage[mathlines]{lineno}

\newtheorem{theorem}{Theorem}
\newtheorem{definition}{Definition}

\usepackage{listings}

\usepackage[usenames,dvipsnames]{xcolor}
\newcommand{\remrkblue}[2]{\textcolor{blue}{\textsc{#1:}}\textcolor{blue}{\textsf{#2}}}
\newcommand{\manfred}[2][]{\remrkblue{Manfred #1}{#2}}
\usepackage{todonotes}

\DeclareMathOperator{\sgn}{sgn}
\DeclareMathOperator{\conv}{conv}
\DeclareMathOperator{\interior}{int}

\title{A SAT attack on higher dimensional Erd\H{o}s--Szekeres numbers\thanks{
		An extended abstract of this article appeared at EuroComb~2021 \cite{Scheucher2021_eurocomb}.
		The author was supported 
		by the DFG Grant SCHE~2214/1-1 ``Approaching Problems from Combinatorics and Geometry using Computer Assistance''.
		We thank Helena Bergold for many valuable comments and Ralf Hoffmann, Sven Grottke, and Norbert Paschedag for the technical support.
	}}
\author{Manfred Scheucher}
\date{}

\begin{document}
	
%	\linenumbers
\author{
	Manfred Scheucher%\inst{1}
}

\maketitle

\begin{center}
	{%\footnotesize
		%\inst{1} 
		Institut f\"ur Mathematik, \\
		Technische Universit\"at Berlin, Germany,\\
		\texttt{lastname@math.tu-berlin.de}
		\\\ \\
	}
\end{center}

%%%%%%%%%%%%%%%%%%%%%%%%%%%%%%%%%%%%%%%%%%%%%%%%%%
%%%%%%%%%%%%%%%%%%%%%%%%%%%%%%%%%%%%%%%%%%%%%%%%%%

\begin{abstract}
	A famous result by Erd\H{o}s and Szekeres (1935) asserts that,
	for all $k,d \in \mathbb{N}$, there is a smallest
	integer $n = g^{(d)}(k)$
	such that every set of at least $n$ points in $\mathbb{R}^d$ in general position 
	contains a \emph{$k$-gon}, that is, 
	a subset of $k$ points which is in convex position.
	In this article, we present a SAT model 
	based on acyclic chirotopes (oriented matroids)
	to investigate 
	Erd\H{o}s--Szekeres numbers in small dimensions.
	To solve the SAT instances we use modern SAT solvers 
	and all our unsatisfiability results are verified using DRAT certificates.
	We show 
	$g^{(3)}(7) = 13$,
	$g^{(4)}(8) \le 13$,
	and
	$g^{(5)}(9) \le 13$,
	which are the first improvements for decades.
	For the setting of \emph{$k$-holes} (i.e., $k$-gons with no other points in the convex hull), 
	where  $h^{(d)}(k)$ denotes the minimum number $n$ 
	such that every set of at least $n$ points in $\mathbb{R}^d$ in general position
	contains a $k$-hole,
	we show
	$h^{(3)}(7) \le 14$, 
	$h^{(4)}(8) \le 13$, and
	$h^{(5)}(9) \le 13$.
	Moreover, all obtained bounds are sharp in the setting of acyclic chirotopes
	and we conjecture them to be sharp also in the original setting of point sets.
	As a byproduct, we verify previously known bounds. In particular, 
	we present the first computer-assisted proof of the
	upper bound $h^{(2)}(6)\le g^{(2)}(9) \le 1717$ by Gerken (2008).
\end{abstract}

%\medskip
%\textbf{Keywords:}
% Erd\H{o}s--Szekeres theorem
% Gerken's empty hexagon theorem
% higher dimensional point set
% acyclic chirotope
% oriented matroid
% $k$-gon
% $k$-hole
% Boolean satisfiability (SAT)
% computer-assisted proof

\goodbreak

\section{Introduction}

A set of points~$S$ in the plane is \emph{in general position} if no three points of $S$ lie on a common line and \emph{in convex position} if all point of $S$ lie on the boundary of the convex hull $\conv(S)$.
The classical Erd\H{o}s--Szekeres Theorem \cite{ErdosSzekeres1935}  asserts that every sufficiently large point set in the plane in general position
%(i.e., no three points on a common line) 
contains a \emph{$k$-gon}, that is, a subset of $k$ points which is in convex position.

\begin{theorem}[{\cite{ErdosSzekeres1935}, The Erd\H{o}s--Szekeres Theorem}]
	For every integer $k \ge 3$, 
	there is a smallest
	integer $n = g^{(2)}(k)$ 
	such that every set of at least $n$ points in general position in the plane
	contains a $k$-gon.
\end{theorem}

Erd\H{o}s and Szekeres showed that $g^{(2)}(k) \le \binom{2k-4}{k-2}+1$  \cite{ErdosSzekeres1935} 
and constructed point sets of size $2^{k-2}$ without $k$-gons \cite{ErdosSzekeres1960},
which they conjectured to be extremal.
There were several improvements of the upper bound in the past decades, each of magnitude $4^{k-o(k)}$, and
in 2016,
Suk showed $g^{(2)}(k) \le 2^{k+o(k)}$ \cite{Suk2017}. Shortly after,
Holmsen et al.\ \cite{HolmsenMPT2020} slightly improved the error term in the exponent and presented a generalization to chirotopes (see Section~\ref{sec:prelim} for a definition).
The lower bound
$g^{(2)}(k) \ge 2^{k-2}+1$ is known to be sharp for $k \le 6$.
The value $g^{(2)}(4)=5$ was determined by Klein in 1933 and
$g^{(2)}(5)=9$ was soon later determined by Makai; for a proof see \cite{KalbfleischKalbfleischStanton1970}, \cite{Bonnice1974} or \cite{MorrisSoltan2000}.
However, several decades have passed until Szekeres and Peters \cite{SzekeresPeters2006} managed to show
$g^{(2)}(6)=17$ by using computer assistance.
While their computer program uses thousands of CPU hours,
we have developed a SAT framework \cite{Scheucher2020_CGTA}
which allows to verify this result within only 2 CPU hours,
and an independent verification of their result using SAT solvers was done by Mari\'c \cite{Maric2019}.

\subsection{\texorpdfstring{Planar $k$-holes}{Planar k-holes}}
\label{sec:6holes}

In the 1970's, Erd\H{o}s \cite{Erdos1978} asked 
whether every sufficiently large point set contains a \emph{$k$-hole},
that is, a $k$-gon with the additional property that 
no other point lies in its convex hull.
In the same vein as $g^{(2)}(k)$, we denote by $h^{(2)}(k)$ the smallest integer
such that every set of at least $h^{(2)}(k)$ points in general position in the plane contains a $k$-hole.
This variant differs significantly 
from the original setting 
as Horton managed to construct arbitrarily large point sets
without 7-holes \cite{Horton1983} 
(see \cite{Valtr1992a} for a generalization called \emph{Horton sets}).
While Harborth \cite{Harborth1978} showed $h^{(2)}(5)=10$,
the existence of \mbox{6-holes} remained open until 2006,
when Gerken \cite{Gerken2008} and Nicol\'as \cite{Nicolas2007} independently showed 
that sufficiently large point sets contain 6-holes. 
We also refer the interested reader to Valtr's simplified proof \cite{Valtr2009}.
Today the best bounds concerning 6-holes are $30 \le h^{(2)}(6)\le 1717$.
The lower bound is witnessed by a set of 29 points without 6-holes
that was found via a simulated annealing approach by Overmars \cite{Overmars2002}.
The upper bound $h^{(2)}(6)\le 1717$ is by Gerken\footnote{
Koshelev \cite{Koshelev2009} claimed that $h^{(2)}(6)\le 463$,
but his proof is incomplete. 
In his 50 pages, written in Russian 
(see \cite{Koshelev2007} for a 3-page extended abstract written in English), 
he classifies how potential counterexamples may look like 
and then the article suddenly stops.
In fact,
in the last sentence of Section~3.1,  Koshelev even writes 
that the completion of the proof will be given in the next publication, which in the author's knowledge did not appear.
},
who made an elaborate case distinction on 34 pages 
to show that every 9-gon -- independent of how many interior points it contains -- yields a 6-hole.
Using the estimate $g^{(2)}(k) \le \binom{2k-5}{k-2}+1$ by T\'oth and Valtr \cite{TothValtr2004}, 
he then concluded $h^{(2)}(6) \le g^{(2)}(9)\} \le  \binom{11}{5}+1 = 1717$.

%\paragraph*{The Erd\H{o}s--Szekeres Theorem in Higher Dimensions}
\subsection{Higher dimensions}
\label{ssec:EST_3space}

The notions \emph{general position} (no $d+1$ points in a common hyperplane), \emph{$k$-gon} (a set of $k$ points in convex position), and \emph{$k$-hole} (a $k$-gon with no other points in the convex hull) 
naturally generalize to higher dimensions, 
and so does the Erd\H{o}s--Szekeres Theorem \cite{ErdosSzekeres1935,DanzerGK1963} (cf.\ \cite{MorrisSoltan2000}).
We denote by 
$g^{(d)}(k)$ and $h^{(d)}(k)$
the minimum number of points in $\mathbb{R}^d$ in general position that guarantee the existence of $k$-gon and $k$-hole, respectively.
%It is conjectured that the Erd\H{o}s--Szekeres numbers $g^{(d)}(k)$ for dimension $d \ge 3$ are subexponential \cite{},
In contrast to the planar case,  the asymptotic behavior of the higher dimensional Erd\H{o}s--Szekeres numbers $g^{(d)}(k)$ remains unknown for dimension $d \ge 3$.
While a dimension-reduction argument by K\'arolyi \cite{Karolyi2001} combined with Suk's bound \cite{Suk2017} shows
\[
g^{(d)}(k) \le g^{(d-1)}(k-1)+1 \le  \ldots \le g^{(2)}(k-d+2)+d-2 \le 2^{(k-d)+o(k-d)}
\]
for $k \ge d \ge 3$,
the currently best asymptotic lower bound 
$g^{(d)}(k) = \Omega(c^{\sqrt[d-1]{k}})$ 
with $c=c(d)>1$
is witnessed by a construction by K\'arolyi and Valtr \cite{KarolyiValtr2003}.
For dimension~3, Füredi conjectured  $g^{(3)}(k) = c^{\Theta(\sqrt{k})}$ (unpublished, cf.\ \cite[Chapter~3.1]{Matousek2002_book}).

\subsection{Higher dimensional holes}
Since
Valtr \cite{Valtr1992b} gave a construction for any dimension $d$ without $d^{d+o(d)}$-holes, generalizing the idea of Horton \cite{Horton1983},
the central open problem about higher dimensional holes is to determine the largest value $k=H(d)$ 
such that every sufficiently large set in $d$-space contains a $k$-hole.
Note that with this notation we have $H(2)=6$ because $h^{(2)}(6)<\infty$ \cite{Gerken2008,Nicolas2007} 
 and $h^{(2)}(7)=\infty$ \cite{Horton1983}.
Very recently Bukh, Chao and Holzman \cite{BukhChaoHolzman2022} presented 
a construction without $2^{7d}$-holes, 
which further improves Valtr's bound and shows $H(d) < 2^{7d}$. 
Another remarkable result is by Conlon and Lim \cite{ConlonLim2021},
who recent generalized \emph{squared Horton sets} \cite{Valtr1992a} to higher dimensions (point sets which are perturbations of $d$-dimensional grids and which do not contain large holes).

On the other hand,
the dimension-reduction argument by K\'arolyi \cite{Karolyi2001}
also applies to $k$-holes, and therefore
\[
h^{(d)}(k) \le h^{(d-1)}(k-1)+1 \le \ldots \le h^{(2)}(k-d+2)+d-2.
\]
This inequality together with $h^{(2)}(6) < \infty$ 
implies that $h^{(d)}(d+4) < \infty$ and hence $H(d) \ge d+4$.
A slightly better bound of $H(d) \ge 2d+1$ was provided by Valtr~\cite{Valtr1992b}, 
who showed $h^{(d)}(2d+1) \le g^{(d)}(4d+1)$.
However, already in dimension~3 the gap between the upper and the lower bound of $H(3)$ remains huge:
while there are arbitrarily large sets without 23-holes~\cite{Valtr1992b}, 
already the existence of 8-holes remains unknown ($7 \le H(3) \le 22$).

\subsection{Precise values}
As discussed before,
for the planar $k$-gons
$g^{(2)}(5)=9$, 
$g^{(2)}(6)=17$, 
$h^{(2)}(5)=10$,
and $g^{(2)}(k) \le \binom{2k-5}{k-2}+1$ are known.
For planar $k$-holes,  
$h^{(2)}(5)=9$, 
$30 \le h^{(2)}(6) \le 463$,
and $h^{(2)}(k)=\infty$ for $k\ge 7$.

%Precise values of $g^{(d)}(k)$ and $h^{(d)}(k)$ have been investigated for certain values of $k$ and~$d$.
While the values $g^{(d)}(k) =h^{(d)}(k) = k$ for $k \le d+1$ and $g^{(d)}(d+2)=h^{(d)}(d+2)=d+3$ are easy to determine %\cite{DanzerGK1963,BisztriczkySoltan1994,Gruenbaum2005}, 
(cf.\ \cite{BisztriczkySoltan1994}),
Bisztriczky et al.\  \cite{BisztriczkySoltan1994,BisztriczkyHarborth1995,MorrisSoltan2000}
showed $g^{(d)}(k) = h^{(d)}(k) = 2k-d-1$ for $d+2 \le k \le \frac{3d}{2}+1$.
This, in particular, determines the values for $(k,d)=(3,5),(4,6),(4,7),(5,7),(5,8)$
and shows $H(d) \ge \lfloor\frac{3d}{2}\rfloor+1$.
%$g^{(3)}(5)=h^{(3)}(5)=9$,
%$g^{(4)}(6)=h^{(4)}(6)=7$, 
%$g^{(4)}(7)=h^{(4)}(7)=9$, 
%$g^{(5)}(7)=h^{(5)}(7)=8$, and
%$g^{(5)}(8)=h^{(5)}(8)=10$.
For $k > \frac{3d}{2}+1$ and $d\ge 3$,
Bisztriczky and Soltan~\cite{BisztriczkySoltan1994}
moreover determined the values
$g^{(3)}(6)=h^{(3)}(6)=9$.
Tables~\ref{tab:high_dim_gons} and~\ref{tab:high_dim_holes} 
summarize the currently best bounds for $k$-gons and $k$-holes in small dimensions.

%%%%%%%%%%%%%%%%%%%%%%%%%%%%%%%%%%%%%%%%%%%%%%%%%%%
\input{table_ES_high_dim.tex}
%%%%%%%%%%%%%%%%%%%%%%%%%%%%%%%%%%%%%%%%%%%%%%%%%%%

\section{Our results}
In this article 
we generalize our SAT framework from \cite{Scheucher2020_CGTA} 
to higher dimensions (see Section~\ref{sec:SAT_framework}).
Our framework models $k$-gons and $k$-holes in terms of acyclic chirotopes,
which generalize point sets in a natural manner (see Section~\ref{sec:prelim}). 
Using this framework, we have been able 
to verify  previously known results and
to prove the following new upper bounds for higher dimensional Erd\H{o}s--Szekeres numbers and for the variant of $k$-holes  in dimensions~3, 4 and~5, which we moreover conjecture to be sharp.
\begin{theorem}
	\label{thm:bounds}
	It holds 
	$g^{(3)}(7) = 13$, 
	$h^{(3)}(7) \le 14$, 
	$g^{(4)}(8) \le  h^{(4)}(8) \le 13$, and
	$g^{(5)}(9) \le  h^{(5)}(9) \le 13$.
\end{theorem}

\goodbreak

To determine satisfiability of the generated SAT instances, 
we use the SAT solver \verb|CaDiCaL|, version~1.0.3 \cite{Biere2019}. 
If an instance is unsatisfiable,
\verb|CaDiCaL| generates an
DRAT certificate which can then
be verified by an independent proof checking tool 
such as \verb|DRAT-trim| \cite{WetzlerHeuleHunt2014}.
Details on time running times and used resources are deferred to Section~\ref{sec:resources}.

Our SAT framework found chirotopes that witness that 
all bounds from Theorem~\ref{thm:bounds} 
are sharp in the more general setting of acyclic chirotopes.
The data is available on our supplemental website~\cite{website_ES_highdim}.
For the chirotope in rank~4 without 7-gons we managed to find a realization.
The explicit coordinates of this set of 12 points in $\mathbb{R}^3$ without 7-gons are:
\begin{linenomath}
\begin{align*}
\{&&(526,446,232)&, &(0,756,64)&, &(612,660,342)&, &(708,638,193) &,&\\
	&&(546,563,134)&, &(616,622,174)&, &(414,0,370)&, &(548,594,151)&,&\\
	&&(884,1334,722)&, &(452,668,180)&, &(587,659,156)&, &(579,692,0)&&\}.
\end{align*}
\end{linenomath}

For the other witnessing chirotopes, however, 
we could not find a realization
but we
conjecture that all bounds from Theorem~\ref{thm:bounds} are sharp in the original setting.
% but we are looking forward to implementing further computer tools so that we can address all those realizability issues.
It is worth noting that finding realizable witnesses 
is a notoriously hard and challenging task as
the problem of deciding realizability is ETR-complete in general (cf.\ Chapters~7.4 and~8.7 in \cite{BjoenerLVWSZ1993}).
Moreover, since only $2^{\Theta(n \log n)}$ of the $2^{\Theta(n^d)}$ rank $d+1$ chirotopes are realizable by point sets in $\mathbb{R}^d$,
it is very likely that a particular witness is not realizable 
while there might exist others which are.

\medskip

We have also used our framework to 
investigate the existence of 8-gons and 8-holes in 3-space.
\verb|CaDiCaL| 
managed to find 
a rank~4 chirotope on 18 elements without 8-gons 
and
a rank~4 chirotope on 19 elements without 8-holes; 
see the supplemental website~\cite{website_ES_highdim}.
The computations took 100 CPU days and 24 CPU days, respectively.
Since it gets harder to find chirotopes without 8-holes as the number of points increase (8-hole-free chirotopes on 18 or less elements can be found within only a few CPU hours),
this is computational evidence that sufficiently large sets in 3-space contain 8-holes.

\medskip

Last but not least, we have used our SAT framework to
verify the previously known bounds
$g^{(2)}(5) \le 9$
(Makai 1935),
$g^{(2)}(6) \le 17$
\cite{SzekeresPeters2006},
$h^{(2)}(5) \le 10$
\cite{Harborth1978},
and
\mbox{$h^{(3)}(6) \le 9$}
\cite{BisztriczkySoltan1994}.
In particular, we present the first computer-assisted proof 
for Gerken's upper bound $h^{(2)}(6)\le g^{(2)}(9) \le 1717$ \cite{Gerken2008}; 
details are deferred to Section~\ref{sec:hexagons}.

\goodbreak
\section{Preliminaries}
\label{sec:prelim}

Let $ S=\{p_1,\ldots,p_n\}$ be a set of $n$ labeled points in $\mathbb{R}^d$
 in general position 
 with coordinates $p_i = (x_{i,1},\ldots,x_{i,d})$.
We assign to each $(d+1)$-tuple ${i_0},\ldots,{i_d}$ a sign to indicate 
whether the $d+1$ corresponding points $p_{i_0},\ldots,p_{i_d}$ are positively or negatively oriented. Formally, we define a mapping $\chi_S : \{1,\ldots,n\}^d \to \{-1,0,+1\}$ with
\[
\chi_S(i_0,\ldots,i_d)
 = \sgn \det 
\begin{pmatrix}
1  & 1  & \ldots & 1  \\
p_{i_0}& p_{i_1}& \ldots & p_{i_d}\\
\end{pmatrix} 
= \sgn \det 
\begin{pmatrix}
1  & 1  & \ldots & 1  \\
x_{i_0,1}& x_{i_1,1}& \ldots & x_{i_d,1}\\
\vdots&\vdots& &\vdots\\
x_{i_0,d}& x_{i_1,d}& \ldots & x_{i_d,d}\\
\end{pmatrix}.
\]

The properties of the mapping $\chi_S$ are captured in the following definition (cf.\ \cite[Definition~3.5.3]{BjoenerLVWSZ1993}).

\begin{definition}[Chirotope]
	\label{def:chirotope}
	A mapping $\chi\colon \{1,\ldots,n\}^r \to \{-1,0,+1\}$ 
	is a \emph{chirotope} of rank~$r$
	if the following three properties are fulfilled:
	\begin{enumerate}[(i)]
		\item \label{def:chirotope1}
		$\chi$ is not identically zero;
		
		\item \label{def:chirotope2}
		for every permutation $\sigma$ and indices $a_1,\ldots,a_r \in \{1,\ldots,n\}$,
		\[
		\chi(a_{\sigma(1)},\ldots,a_{\sigma(r)}) = \sgn(\sigma) \cdot\chi(a_1,\ldots,a_r) ;
		\]
		
		\item \label{def:chirotope3}
		for all indices $a_1,\ldots,a_r,b_1,\ldots,b_r \in \{1,\ldots,n\}$,		
		\begin{linenomath}
		\begin{align*}
		\text{if}\quad
		\chi(b_i,a_2,\ldots,a_r)\cdot \chi(b_1,\ldots,b_{i-1},a_1,b_{i+1},\ldots,b_r) \ge 0& \quad \text{holds for all $i=1,\ldots,r$} \\
		\text{then we have}\quad  \chi(a_1,\ldots,a_r) \cdot \chi(b_1,\ldots,b_r) \ge 0&. 
		\end{align*}
		\end{linenomath}
	\end{enumerate}
\end{definition}

It is not hard to verify that the mapping $\chi_S$ is a chirotope of rank $r=d+1$.
The first item of Definition~\ref{def:chirotope} is fulfilled because the point set~$S$ is in general position 
and therefore not all points lie in a common hyperplane. 
The second first item is fulfilled because,
by the properties of the determinant, we have
\[
\det(a_{\sigma(1)},\ldots,a_{\sigma(r)}) = \sgn(\sigma) \cdot \det(a_1,\ldots,a_r)
\]
for any $r$-dimensional vectors $a_1,\ldots,a_r$ and any permutation $\sigma$ of the indices $\{1,\ldots,r\}$.
Since we can consider the homogeneous coordinates $(1,p_1),\ldots,(1,p_n)$ of our $d$-dim\-ension\-al point set $S$ as $(d+1)$-dimensional vectors,
the above relation also has to be respected by $\chi_S$.
%where $\sgn(\sigma)$ denotes the sign of the permutation~$\sigma$,
%and the same applies to~$\chi_S$.
To see that $\chi_S$ also fulfills the third item,
recall that the well-known Gra{\ss}mann--Pl\"ucker relations (see e.g.\ \cite[Chapter~3.5]{BjoenerLVWSZ1993}) 
assert that any $r$-dimensional vectors $a_1,\ldots,a_r,b_1,\ldots,b_r$ fulfill%
\footnote{
	The Gra{\ss}mann--Pl\"ucker relations can be derived 
	 using  Linear Algebra basics 
	as outlined: 
	Consider the vectors $a_1,\ldots,a_r,b_1,\ldots,b_r$ 
	as an $r \times 2r$ matrix and
	apply row additions to obtain the echelon form.
	If the first $r$ columns form a singular matrix,
	then $\det(a_1,\ldots,a_n) = 0$ and 
	both sides of the equation vanish by a simple column multiplication argument.
	Otherwise, we can assume 
	%(also due to a column multiplication argument) 
	that the first $r$ columns form an identify matrix.
	Since the determinant is invariant to {row additions},
	none of the terms in the Gra{\ss}mann--Pl\"ucker relations 
	is effected during the transformation,
	and the statement then follows from Laplace expansion. 
}
\begin{linenomath}
\begin{align*}
\det(a_1,\ldots,a_r) \cdot& \det(b_1,\ldots,b_r)
\\
&= \sum_{i=1}^{r} \det(b_i,a_2,\ldots,a_r)\cdot \det(b_1,\ldots,b_{i-1},a_1,b_{i+1},\ldots,b_r).
\end{align*}
\end{linenomath}
In particular, if all summands on the right-hand side are non-negative
then also the left-hand side must be non-negative.

In this article, 
%we are mainly interested in chirotopes which come from point sets as 
%we look forward to 
we formulate Erd\H{o}s--Szekeres-type problems in terms of chirotopes.
In a general chirotope, we can reverse the sign of all $r$-tuples which contain a fixed index~$i$
and again obtain a valid chirotope.
However,
if we apply such a reversal to a chirotope, which is induced by a point set,
the obtained chirotope might not be induced by any point set.
More specifically, if the points $p_1,\ldots,p_r$ determine a positively oriented simplex 
(i.e., $\chi(1,2,\ldots,r) = +$),
then another point $p_{r+1}$ might lie inside the simplex 
(i.e., $\chi(1,2,\ldots,i-1,r+1,i+1,\ldots,r) = +$ for all $i=1,\ldots,r$),
but it cannot fulfill $\chi(1,2,\ldots,i-1,r+1,i+1,\ldots,r) = -$ for all $i=1,\ldots,r$.
Chirotopes with this property 
are called \emph{acyclic} 
(see e.g.\ \cite{BjoenerLVWSZ1993}).
It is worth noting besides chirotopes there exist 
several cryptomorphic axiomatizations for oriented matroids 
and that acyclic chirotopes
have been investigated under several different names:
\emph{abstract order types} (cf.\ \cite{AichholzerAurenhammerKrasser2001,Krasser2003,AichholzerKrasser2006,Finschi2001,FinschiFukuda2002}),
\emph{CC systems} (cf.\ \cite{Knuth1992}), or
\emph{pseudoconfigurations of points} (cf.\ \cite{BjoenerLVWSZ1993}).
%or 
%\emph{pseudolinear drawings of the complete graph~$K_n$} (cf.\ \cite{ArroyoMQRS2018}).
We refer the interested reader to the homepage of (acyclic) oriented matroids \cite{FinschiDBOM};
for non-degenerate rank~3 case, see \cite{AichholzerOTDB}.
It is also worth noting that 
for any acyclic chirotope of rank~3
there exists a reordering of the elements 
such that it becomes a \emph{signotope} of rank~3 (cf.\ \cite{FelsnerWeil2001}).

In Section~\ref{sec:SAT_framework}
we present a SAT model for (acyclic) chirotopes.
While the axioms from Definition~\ref{def:chirotope} require $\Theta(n^{2r})$ constraints,
we can significantly reduce this number to $\Theta(n^{r+2})$ by using an axiom system based on the 3-term Gra{\ss}mann--Pl\"ucker relations.

\begin{theorem}[3-term Gra{\ss}mann--Pl\"ucker relations, {\ \cite[Theorem~3.6.2]{BjoenerLVWSZ1993}}]
	
	\label{thm:3term}
	A mapping $\chi\colon \{1,\ldots,n\}^r \to \{-1,0,+1\}$ 
	is a {non-degenerate chirotope} of rank~$r$
	if the following three properties are fulfilled:
	\begin{enumerate}[(i)]
%		\item \label{def:chirotope1}
%		for any $r$ distinct indices $a_1,\ldots,a_r \in \{1,\ldots,n\}$,
%		\[
%		\chi(a_1,\ldots,a_r) \neq 0;
%		\]
		
		\item \label{def:3term1}
		for every $r$ distinct indices $a_1,\ldots,a_r \in \{1,\ldots,n\}$,
		\[ \chi(a_1,\ldots,a_r) \neq 0;
		\]
		
		\item \label{def:3term2}
		for every permutation $\sigma$ and indices $a_1,\ldots,a_r \in \{1,\ldots,n\}$,
		\[
		\chi(a_{\sigma(1)},\ldots,a_{\sigma(r)}) = \sgn(\sigma) \cdot\chi(a_1,\ldots,a_r) ;
		\]
		
		\item \label{def:3term3}
		for any $a_1,\ldots,a_r,b_1,b_2 \in \{1,\ldots,n\}$,
		\[
		\begin{array}{ll}
		\text{if}\quad&
		\chi(b_1,a_2,\ldots,a_r)\cdot \chi(a_1,b_{2},a_{3},\ldots,a_r) \ge 0\\
		\text{and}\quad&
		\chi(b_2,a_2,\ldots,a_r)\cdot \chi(b_1,a_{1},a_{3},\ldots,a_r) \ge 0\\
		\text{then}\quad  &\chi(a_1,a_2,\ldots,a_r) \cdot \chi(b_1,b_2,a_3,\ldots,a_r) \ge 0. 
		\end{array}
		\]
	\end{enumerate}
\end{theorem}

\subsection{Gons and holes}
\label{ssec:GonsAndHoles}

Carath\'eodory's theorem asserts that a $d$-dimensional point set is in convex position if and only if all $(d+2)$-element subsets are in convex position. 
Now that a point $p_{i_{d+1}}$ lies in the convex hull of $\{p_{i_0},\ldots,p_{i_d}\}$ 
if and only if $\chi(i_0,\ldots,i_d) = \chi(i_0,\ldots,i_{j-1},i_{d+1},i_{j+1},\ldots,i_d)$ holds for every $j \in \{0,\ldots,d\}$,
we can fully axiomize $k$-gons and $k$-holes
solely using the information of the chirotope, 
that is, the relative position of the points. 
Note that $\chi(i_0,\ldots,i_d) = \chi(i_0,\ldots,i_{j-1},i_{d+1},i_{j+1},\ldots,i_d)$ holds if and only if 
the two points $p_{i_j}$ and $p_{i_{d+1}}$
lie on the same side of
the hyperplane determined by $\{p_{i_0}\ldots,p_{i_{j-1}},p_{i_{j+1}},\ldots,p_{i_d}\}$.

\section{The SAT framework}
\label{sec:SAT_framework}

For the proof of Theorem~\ref{thm:bounds}, we proceed as following:
To show $g^{(d)}(k) \le n$ or $h^{(d)}(k) \le n$, respectively,
assume towards a contradiction that 
there exists a set $S$ of $n$ points in $\mathbb{R}^d$ in general position,
which does not contain any $k$-gon or $k$-hole, respectively.
The point set $S$ induces an acyclic chirotope $\chi$ of rank~$d+1$,  
which can be encoded using $n^{d+1}$ Boolean variables.
The chirotope $\chi$ fulfills the 
%$\Theta(n^{2d+2})$ conditions from Definition~\ref{def:chirotope},
%which we can encode as clauses.
%Moreover, we can significantly reduce the number of clauses because, 
%due to the \emph{3-term Gra{\ss}mann--Pl\"ucker relations} 
%(Theorem~3.6.2 in \cite{BjoenerLVWSZ1993}),
$\Theta(n^{d+3})$ conditions from Theorem~\ref{thm:3term},
which we can encode as clauses.
% Remark: Chirotope axioms would give $\Theta(n^{2d+2})$ clauses.

Next, we introduce auxiliary variables for all $(d+2)$-tuples $i_0,\ldots,i_{d+1} \in \{1,\ldots,n\}$
to indicate 
%\begin{itemize}
%\item 
whether the hyperplane determined by $\{p_{i_0},\ldots,p_{i_{d-1}}\}$ separates the two points $p_{i_d}$ and $p_{i_{d+1}}$,
and 
%\item 
whether the point $p_{i_{d+1}}$ lies in the convex hull of $\{p_{i_0},\ldots,p_{i_d}\}$.
%\end{itemize}
As discussed in Section~\ref{ssec:GonsAndHoles},
the values of these auxiliary variables are fully determined by the chirotope variables.

Using these auxiliary variables 
we can formulate $\binom{n}{k}$ clauses to assert that there are no $k$-gons in~$S$:
Among every subset $X \subset S$ of size $|X|=k$ there is
at least one point $p \in X$ which is contained in the convex hull of $d+1$  points of~$X \setminus \{p\}$.
Moreover, since the convex hull of $X \setminus \{p\}$ can be triangulated in a way such that all simplices of the triangulation contain the leftmost point of $X \setminus \{p\}$,
the point $p$ lies in a 
simplex determined by the leftmost and $d$ further points from $X \setminus \{p\}$.
%simplex determined by $X \setminus \{p\}$
%which moreover is incident to the leftmost point of $X \setminus \{p\}$.
To assert that there are no $k$-holes in~$S$, we can proceed in a similar manner:
For every subset $X \subset S$ of size $|X|=k$ there is
at least one point $p \in S$ which is contained in the convex hull of $d+1$  points of~$X \setminus \{p\}$. 
And again, $p$ lies in a simplex determined by the leftmost and $d$ further points from $X \setminus \{p\}$.

Altogether, we can now create a
Boolean satisfiability instance that is satisfiable if and only if 
there exists a rank $d+1$ chirotope on $n$ elements 
without $k$-gons or $k$-holes, respectively.
%The instance has $O(n^{\max(d+3,k)})$ clauses in $O(n^{d+2})$ variables.
If the instance is unsatisfiable,
no such chirotope (and hence no point set) exists,
and we have $g^{(d)}(k) \le n$ or $h^{(d)}(k) \le n$, respectively.

The python program for creating the instances and further technical information is available on our supplemental website~\cite{website_ES_highdim}.

\section{Running times and resources}
\label{sec:resources}

All our computations were performed on a single CPU. 
Since the computations 
required a lot of resources (RAM and disk space), 
%than available on a standard computer/laptop,
%(especially for verifying the unsatisfiability certificates), 
we made use of the computing cluster from the Institute of Mathematics at TU Berlin. 
Below are the running times and used resources for each computation.

\goodbreak

\begin{itemize}
	\item
	 $g^{(3)}(7) \le 13$: 
	 The size of the instance is about 245~MB
	 and \verb|CaDiCaL| managed to prove unsatisfiability in about 2~CPU~days.
	 Moreover, the unsatisfiability certificate created by \verb|CaDiCaL| is about 39~GB
	 and the  
	 \verb|DRAT-trim| verification took about 1~CPU~day.

	 \item
	 $h^{(3)}(7) \le 14$:
	 The size of the instance is about 433~MB
	 and \verb|CaDiCaL| (with parameter \verb|--unsat|) 
	 managed to prove unsatisfiability in about 19~CPU~days.
     Moreover, the unsatisfiability certificate created by \verb|CaDiCaL| is about 314~GB
     and the \verb|DRAT-trim| verification took about 12~CPU~days.
     
	 \item
	 $h^{(4)}(8) \le 13$:
	 The size of the instance is about 955~MB
	 and \verb|CaDiCaL| managed to prove unsatisfiability in about 7~CPU~days.
	 Moreover, the unsatisfiability certificate created by \verb|CaDiCaL| is about 297~GB
	 and the \verb|DRAT-trim| verification took about 6~CPU~days.

	 \item
	 $h^{(5)}(9) \le 13$:
	 The size of the instance is about 4.2~GB
	 and \verb|CaDiCaL| managed to prove unsatisfiability in about 3~CPU~days.
	 Moreover, the unsatisfiability certificate created by \verb|CaDiCaL| is about 117~GB
	 and the  
	 \verb|DRAT-trim| verification took about 3~CPU~days. 
	 
\end{itemize}

\section{Verification of previous results}
\label{sec:hexagons}

We have also used our SAT framework to
verify the previously known bounds.
For the unsatisfiability of $g^{(2)}(6) \le 17$ \cite{SzekeresPeters2006},
\verb|CaDiCaL|  takes about 10~CPU minutes to create a DRAT-certificate.
The generated DRAT-certificate has about 668~MB and can be verified by \verb|DRAT-trim| within 12~CPU minutes.
The instances 
$g^{(2)}(5) \le 9$
(Makai 1935),
$h^{(2)}(5) \le 10$
\cite{Harborth1978},
and
$h^{(3)}(6) \le 9$
\cite{BisztriczkySoltan1994}
can be solved and verified within only a few CPU seconds 
and the respective DRAT-certificates have only a few~MB.
It is worth noting that our new framework performs much faster on proving $g^{(2)}(6) \le 17$ than our old one from \cite{Scheucher2020_CGTA} (see also \cite{Maric2019})
because we now make use of a triangulation (see Section~\ref{sec:SAT_framework}) 
which significantly reduces the number of literals in certain clauses and which apparently reduces the computation time.

\subsection{Existence of planar 6-holes}

Last but not least,
%Here 
we present the a computer-assisted proof 
for the existence of 6-holes in planar point sets. 
To our knowledge, this is the first alternative proof 
to Gerken's 34 pages of case distinction~\cite{Gerken2008}.
Note that the bound even applies to the more general setting of acyclic chirotopes of rank~3.

\begin{theorem}[\cite{Gerken2008}]
	\label{thm:hexagons}
$h^{(2)}(6)\le g^{(2)}(9)$.
\end{theorem} 

\begin{proof}
Suppose towards a contradiction that 
there exists a set $S$ of $g^{(2)}(9)$ points in the plane which does not contain a 6-hole.
By definition of $g^{(2)}(9)$, 
there exists a subset $A \subseteq S$ of $|A|=9$ points which determine a 9-gon.
We now restrict our attention to $S' = S \cap \conv(A)$.

Let $B$ denote the extremal points of $S' \cap \interior( \conv(A) )$,
let $C$ denote the extremal points of $S' \cap \interior( \conv(B) )$,
and let $D = S' \cap \interior ( \conv(C) )$. 
Here $\interior( X )$ denotes the interior of~$X$.
In other words,
$A$, $B$, and $C$ span the first, second, and third convex hull layer 
of $\conv(S')$, respectively,
and $D$ contains all points that lie on or inside the fourth convex hull layer.

Lemma~1 from \cite{Valtr2009}
now asserts that $D = \emptyset$, as otherwise $S'$ (and thus $S$) would contain a 6-hole.
Since the points of $C$ form a $|C|$-hole,
it clearly holds $|C| \le 5$.
%An analogous argument shows that $|C| =0  \Rightarrow |B| \le 5$,
%and since $|A|=9$, it holds $|B| \ge 1$.
Furthermore, we may assume that $A$ is a \emph{minimal 9-gon in~$S'$},
that is, $A$ is the only 9-gon in~$S'$,
as otherwise we could iteratively choose 9-gons inside~$S'$ 
which contain fewer 9-gons until we end up with a minimal one.
Hence, we have $|B| \le 8$,
and in total 
\[
|S'| \le |A|+|B|+|C| \le 9+8+5 =22.
\]

By slightly modifying our framework from Section~\ref{sec:SAT_framework}, 
we can assert that the 9~points of $A$ form a 9-gon 
and all other points of $S' \setminus A$ lie in the interior of $\conv(A)$.
Using \verb|CaDiCaL|  
we managed to verify that, for every $n = 9,\ldots,22$, 
there exists no set of $S'$ with $|S'|=n$.
This contradiction shows that 
every a set of $g^{(2)}(9)$ points in the plane contains a 6-hole.
\end{proof}

The total size of all 14 instances in the proof of Theorem~\ref{thm:hexagons} 
is about 3.8~GB.
The total computation time of \verb|CaDiCaL| and \verb|DRAT-trim|
to create and verify the DRAT-certificates for the instances is about 1 CPU hour.
The total size of the generated DRAT-certificates is about 6.2~GB.
The python program for creating the instances and further technical information is available on our supplemental website~\cite{website_ES_highdim}.

%none of the
%18~cases
%\begin{itemize}
%	\item 
%	$3 \le |B| \le 8$ and $3 \le |C| \le 5$
%\end{itemize}
%or the five cases
%\begin{itemize}
%	\item 
%	$1 \le |B| \le 5$ and $ |C| = 0$
%\end{itemize}
%can occur.

\section{Discussion}
We have presented a SAT framework to investigate problems on (acyclic) chirotopes. 
Using modern SAT solvers we
obtained new results for Erd\H{o}s--Szekeres-type problems in dimensions $d=3,4,5$,
and we verified various previously known results.
In particular, we have verified Gerken's proof 
for the existence of 6-holes in planar point sets.
Based on computation evidence 
(acyclic chirotopes on $n = 29$ points without 6-holes can be found within a few minutes, 
but for $n=30$ the computations do not terminate within months),
we conjecture
that every set of 30 points in the plane contains at least one 6-hole, that is, $h^{(2)}(6) = 30$.

Since some of the unsatisfiability certificates generated in course of this article grew very large,
it would be interesting to further optimize the SAT model by breaking further symmetries so that the solver can terminate faster and the obtained DRAT certificates become smaller.

Last but not least we want to mention that 
our framework can also be used to tackle 
other problems on higher dimensional point sets or oriented matroids. 
By slightly adapting our model,
we managed to answer a Tverberg-type question by Fulek et al.\ (see Section~3.2 in \cite{FulekGKVW2019} for more details)
and we managed to find a chirotope which is a contact representation of particular hypergraph, partially answering a question by Evans et al.~\cite{EvansRSSW2019}.
\iffalse
To be more specific, Evans et al.\ asked whether 
every Steiner triple system $S(2, 3, n)$ with $n \ge 13$ 
admits a non-crossing drawing using triangles in 3-space.
We found a rank~4 chirotope representing one of the two Steiner triple systems for $S(2, 3, 13)$,
but the realizability of this chirotope remains unknown.
We expect that the solver can also find chirotope representations of the 80 systems $S(2, 3, 15)$, or 
show that they are not representable (not even representable by chirotopes). 
%Depending on the outcome, this will either partially or negatively answer the question from \cite{EvansRSSW2019}.
\fi

{
\small
\bibliographystyle{alphaabbrv-url}
\bibliography{references}
}

\end{document}

%% file: table_ES_high_dim.tex
\begin{table}
%\begin{wraptable}{r}{0.45\textwidth}
%\footnotesize
	\centering
	\begin{tabular}{ r | rrrrrrrrrrrr}
	 			&$k=4$	&$5$	&$6$	&$7$	&$8$ 	&$9$	&$10$	&$11$  \\
		\hline
		$d=2$  	&$5$	&$9$	&$17$	&	&	&	&	&\\
		$3$		&$4$	&$6$	&$9$	&$13$*	&	&	&	&\\
		$4$		&$4$	&$5$	&$7$	&$9$	&$ \le 13$*	&	&	&\\
		$5$		&$4$	&$5$	&$6$	&$8$	&$10$	&$ \le 13$*	&	&\\
		$6$		&$4$	&$5$	&$6$	&$7$	&$9$	&$11$	&$13$	&\\
	\end{tabular}
	\caption{Known values and bounds for $g^{(d)}(k)$.  
		Entries marked with a star (*) are new. Entries  left blank are
		upper-bounded  by the estimate $g^{(2)}(k) \le \binom{2k-5}{k-2}+1$ \cite{TothValtr2004} and the dimension-reduction argument \cite{Karolyi2001}.}
	\label{tab:high_dim_gons}
%\end{wraptable}
\end{table}

\begin{table}
	%\begin{wraptable}{r}{0.45\textwidth}
	%\footnotesize
	\centering
	\begin{tabular}{ r | rrrrrrrrrrrrr}
				&$k=4$	&$5$	&$6$	&$7$	&$8$ 	&$9$	&$10$	&$11$	&$12$	&$13$  \\
		\hline
		$d=2$  	&$5$	&$10$	&$30..463$	&$\infty$	&$\infty$	&$\infty$	&$\infty$	&$\infty$	&$\infty$	&$\infty$\\
		$3$		&$4$	&$6$	&$9$	&$ \le 14$*	&?	&?	&?	&?	&?	&?\\
		$4$		&$4$	&$5$	&$7$	&$9$	&$ \le 13$*	&!	&?	&?	&?	&?\\
		$5$		&$4$	&$5$	&$6$	&$8$	&$10$	&$ \le 13$*	&!	&!	&?	&?\\
		$6$		&$4$	&$5$	&$6$	&$7$	&$9$	&$11$	&$13$	&!	&!	&!\\
	\end{tabular}
	\caption{Known values and bounds for $h^{(d)}(k)$. 
		Entries marked with a star (*) are new.
		Entries marked with an exclamation mark (!) are finite because of the estimate $h^{(d)}(2d+1) \le g^{(d)}(4d+1)$ \cite{Valtr1992b}. 
		Entries marked with a question mark (?) are not known to be finite.
	}
	\label{tab:high_dim_holes}
	%\end{wraptable}
\end{table}